%% file: recur.tex
\documentclass{amsart}  
\usepackage{amsmath}
\usepackage{amsthm}
\usepackage{amsfonts,amssymb}
\usepackage[foot]{amsaddr}

\usepackage{graphicx}

\usepackage{color}

\newtheorem{theorem}{Theorem}
\newtheorem{lemma}{Lemma}
\newtheorem{proposition}{Proposition}
\newtheorem{corollary}{Corollary}

\theoremstyle{definition}

\newtheorem{example}{Example}

\newtheorem{remark}{Remark}

\renewenvironment{proof}[1][]{\par\vskip12pt\noindent%
     {\it Proof\ifx#1==\/:\else{ #1\/:}\newline\fi\ }\bgroup}
     {\egroup\quad\strut\hfill$\blacksquare$\vskip12pt}

\newcommand{\ic}{\ensuremath{i}}

\def\idty{{\leavevmode\rm 1\mkern -5.4mu I}} 
\def\Rl{{\mathbb R}}\def\Cx{{\mathbb C}}
\def\Ir{{\mathbb Z}}\def\Nl{{\mathbb N}}

\def\norm #1{\Vert #1\Vert}

\def\sign{{\mathop{\rm sign}\nolimits}\,}

\def\bra #1{\langle #1\vert}
\def\ket #1{\vert #1\rangle}
\def\braket #1#2{\langle #1 \vert #2\rangle}
\def\ketbra #1#2{\vert #1\rangle \langle #2\vert}
\def\kettbra#1{\ketbra{#1}{#1}}

\def\abs#1{\vert#1\vert}

\def\Ell#1{{\mathcal L^{#1}}}
\def\inv{^{-1}}

\def\HH{{\mathcal H}}\def\KK{{\mathcal K}}

\def\Ut{{\widetilde U}}
\def\RSJK{R^{\rm SJK}}  
\def\tauSJK{\tau^{\rm SJK}}  
\def\qSJK{q^{\rm SJK}}  
\def\Gt{{\widetilde G}}
\def\RE{\Re e} 
\def\hata{{\widehat a}}
\def\muh{{\widehat\mu}}
\def\cyclic#1{(U^\Ir#1)^{\perp\perp}} 
\def\taudoof{\tau^{\rm Polya}}
\def\vartau{V_\tau}

\newcommand{\WSp}{\ensuremath{\ell^2(\Ir^s)\otimes \KK}}

\begin{document}
\title{Recurrence for discrete time unitary evolutions}
\author{F. A. Gr\"unbaum$^1$}
\address{$^1$
Department of Mathematics \\ University of California \\ Berkeley CA 94720}
\author{L. Vel\'azquez$^2$}
 \address{$^2$ Departamento de Matem\'{a}tica Aplicada and IUMA  \\ Universidad de Zaragoza \\ Mar\'{\i}a de Luna 3, 50018 Zaragoza, Spain}
\author{ A.H. Werner$^3$}

\author{ R.F.Werner$^3$}
\address{$^3$ Institut f\"ur Theoretische Physik \\ Leibniz Universit\"at Hannover \\ Appelstr. 2, 30167 Hannover, Germany}

\email{albertogrunbaum@yahoo.com,velazque@unizar.es} \email{albert.werner@itp.uni-hannover.de,reinhard.werner@itp.uni-hannover.de}

\maketitle
\begin{abstract}
We consider quantum dynamical systems specified by a unitary operator $U$ and an initial state vector $\phi$. In each step the unitary is followed by
a projective measurement checking whether the system has returned to the initial state. We call the system recurrent if this eventually happens with
probability one. We show that recurrence is equivalent to the absence of an absolutely continuous part from the spectral measure of $U$ with respect
to $\phi$. We also show that in the recurrent case the expected first return time is an integer or infinite, for which we give a topological
interpretation. A key role in our theory is played by the first arrival amplitudes, which turn out to be the (complex conjugated) Taylor coefficients
of the Schur function of the spectral measure. On the one hand, this provides a direct dynamical interpretation of these coefficients; on the other
hand it links our definition of first return times to a large body of mathematical literature.
\end{abstract}

\section{Introduction}
One of the standard questions addressed in the theory of Markov chains is the distinction of transient and recurrent processes. By definition, a
Markov chain is called recurrent, if almost every path comes back to its starting point, otherwise transient. In this paper we discuss an analogous
question for unitary quantum dynamics. Immediately, this presents us with a problem, which is typical for many generalizations of classical concepts
to the quantum world: The definition as given for Markov chains clearly requires some monitoring of the process: we have to check after every step
whether the particle has returned. But this monitoring, if it is to give any non-trivial information about the system, necessarily changes the
dynamics. Therefore, there are two options: We can either try to reformulate the problem in such a way that the monitoring is not needed, or else we
include the monitoring into the description. An approach of the first kind has recently been put forward by \v{S}tefa\v{n}\'ak, Jex and Kiss
\cite{Stefanak} on the basis of a formula, which P\'olya derived for Markov chains, and which is only based on the return probabilities, i.e., the
probability for the particle to be back {\it at} a given time, without any conditions on the path leading there. We call unitaries with specified
initial state satisfying the criterion \cite{Stefanak} {\it SJK-recurrent}. In this paper we choose the second option. Our notion of recurrence
implies SJK-recurrence, but not conversely. We refrain from comparing our definition with a host of other ones in the literature.

A common feature of these approaches is that the recurrence criterion depends only on the unitary operator $U$ and the initial state $\phi$, and
hence only on the scalar measure $\mu(du)=\braket\phi{E(du)\phi}$ on the unit circle, which is obtained from the spectral measure $E$ of $U$ (so
$U=\int uE(du)$). Ultimately, both criteria have a simple expression in terms of properties of $\mu$. This connects the recurrence problem to a rich
mathematical literature about measures on the unit circle and their orthogonal polynomials (the canonical review of the field is
\cite{Simon1,Simon2}). The first return probabilities in our approach are the squared moduli of the Taylor coefficients of the so-called Schur
function \cite{Schur1917} of the measure, which so far did not seem to have a direct dynamical interpretation. Our main result is that the process is
recurrent iff the Schur function is ``inner'', i.e., has modulus one on the unit circle. Furthermore, we show that the winding number of this
function has the direct interpretation as the expected time of first arrival, which is hence an integer.

We came to this investigation through our respective work on quantum walks \cite{timeRandom,CGMV2010,CMGVI}. Quantum walks which are basically
translation invariant (up to a local perturbation) are not recurrent, but the analysis of first returns is by no means limited to deciding
recurrence. We show how the methods developed in the quantum walk context can be used to get explicit expressions in simple cases.

Our paper is organized as follows. In Sect.~\ref{sec:rec} we recapitulate the classical idea of recurrence and the quantum version by
\cite{Stefanak}, and then propose our alternative based on monitoring the dynamics. In Sect.~\ref{sec:proof1} we give a spectral
characterization of this notion (Theorem~\ref{thm:crit}) and, for the sake of comparison, a similar characterization for the SJK-recurrence. The
determination of the expected first return time is in Theorem~\ref{thm:time} of Sect.~\ref{sec:time}: the main result is that the
expectation of the first return time is, surprisingly, an integer or infinite. Again detailed results on the corresponding SJK-quantity are provided
for comparison. Sect.~\ref{sec:finite} treats in more detail the finite dimensional case and shows, in particular, how the variance of
the first return time diverges when this integer value changes. Examples from the field of quantum walks are presented in Sect.~\ref{sec:as}. We
provide two rather different methods for computing return probabilities, which however, have a large overlap of applicability. We close with  an
outlook on some examples of singular spectrum and natural generalizations of our basic notion of recurrence.

\input{recurN.tex}  
\input{recurP.tex}  
\input{recurT.tex}  
\input{recurF.tex}  
\input{recurA.tex}  
\input{recurS.tex}  

\section*{Acknowledgements}
We are grateful for the hospitality of Centro de Ciencias de Benasque Pedro Pascual in Benasque (Spain), where our collaboration was initiated. A.H.
Werner and L. Vel\'azquez want to thank the organizers of a Quantum Information workshop in the same center, where further parts of this work were
done.

A.H. Werner and R.F. Werner acknowledge support from the Deutsche For\-schungsgemeinschaft (Grant Forschergruppe 635) and the EU project COQUIT.

F.A. Gr\"{u}nbaum acknowledges support from the Applied math. Sciences subprogram of the Office of Energy Research, USDOE, under Contract
DE-AC03-76SF00098.

The work of L. Vel\'azquez was partly supported by the research projects MTM2008-06689-C02-01 and MTM2011-28952-C02-01 from the Ministry of Science
and Innovation of Spain and the European Regional Development Fund (ERDF), by Departamento de Ciencia, Tecnolog\'{\i}a y Universidad from Gobierno de
Arag\'on, by Project E-64 of Diputaci\'on General de Arag\'on and by Instituto Universitario de Matem\'atica Aplicada and Departamento de
Matem\'atica Aplicada from Universidad de Zara\-goza.

\bibliographystyle{abbrv}
\bibliography{recur}

\end{document}

%% file: recurN.tex
\section{Notions of recurrence}\label{sec:rec}
In this section we compare our notion of recurrence to the one introduced by \v{S}tefa\v{n}\'ak, Jex and Kiss \cite{Stefanak}, which avoids
introducing explicit monitoring steps. Their idea is based on the theory of Markov chains, particularly the idea of recurrence as initiated by
P\'olya \cite{Polya}.  We briefly review this classical background \cite{Karlin} here, partly to improve the understanding of \cite{Stefanak}, but
chiefly as a basis to point out close analogies between the classical return probabilities and return amplitudes arising later in our paper. After
introducing our proposal we provide an example demonstrating that the two approaches may come to opposite conclusions.
\subsection{The classical case}\label{sec:classNotion}
A (discrete time) Markov chain on a countable state space $X$ is defined by its transition probability matrix $P_{xy}$, where $x,y\in X$. From the
interpretation of $P_{xy}$ as the probability for the process to move from $x$ to $y$ in one step it is clear that $P_{xy}\geq0$ and
$\sum_yP_{xy}=1$. Then the Markov property entails that the probability for going from $x$ to $y$ along intermediate steps $(x_1,\ldots,x_{n-1})$ is
$P_{x,x_1}P_{x_1,x_2}\cdots P_{x_{n-1},y}$. If we do not care for the intermediate steps these variables are summed over, resulting in the matrix
power  $P^n_{xy}$. Now we single out a point $0\in X$. It will be the initial point for all paths we consider, and we are interested in the returns
to $0$. The probability to return in $n$ steps is obviously
\begin{equation}\label{pnMarkov}
    p_n=P^n_{00}=\sum_{x_1,\ldots,x_{n-1}}P_{0,x_1}P_{x_1,x_2}\cdots P_{x_{n-1},0}.
\end{equation}
We would also like to know the probability $q_n$ for this same transition subject to the condition that no intermediate step is $0$. That is,
\begin{equation}\label{qnMarkov}
    q_n=\sum_{\substack{x_1,\ldots,x_{n-1}\\
                 \mbox{all}\neq0}}P_{0,x_1}P_{x_1,x_2}\cdots P_{x_{n-1},0}
\end{equation}
These numbers are the probability law for the random variable {\it first return time}. Clearly, \eqref{qnMarkov} is the sum over just a subset of the
terms in \eqref{pnMarkov}. If we split the sum \eqref{pnMarkov} according to the highest index $k$ at which $x_k=0$, we get a sum in which the
summation indices $x_\ell$ with $\ell<{k}$ are unconstrained, but those with $\ell>k$ must be non-zero. In other words
\begin{equation}\label{pnConvolve}
    p_n=\sum_{k=0}^n p_kq_{n-k},
\end{equation}
where for the endpoints we use the natural conventions $p_0=1$ and $q_0=0$. This formula is not valid for $n=0$ (where it would give the wrong
equation $1=p_0=p_0q_0=0$). Therefore, if we introduce the generating functions
\begin{equation}\label{gfMarkov}
    \widehat p(z)=\sum_{n=0}^\infty p_nz^n\quad \mbox{and }\quad
    \widehat q(z)=\sum_{n=0}^\infty q_nz^n
\end{equation}
and sum over \eqref{pnConvolve} (the $n=0$ term has to be treated separately) we get
\begin{equation}\label{qhat}
    \widehat p(z)=1+ \widehat p(z)\widehat q(z).
\end{equation}
Hence the {\it return probability} or, more precisely, the probability to eventually return to the initial state $0$,
is now
\begin{equation}\label{siMarkov}
    R^C = \sum_nq_n = \widehat q(1)
        = 1-\frac1{\widehat p(1)} = 1-\frac1{\sum_n p_n}.
\end{equation}
Therefore the Markov chain is recurrent from the point $0$ if and only if $R^C=1$, or $\widehat p(1)=\sum_np_n=\infty$.

\subsection{Recurrence without monitoring}
Now the criterion for recurrence proposed by \cite{Stefanak} is to apply \eqref{siMarkov} to the sequence of quantum return probabilities
\begin{equation}\label{pn}
    p_n=\abs{\braket\phi{U^n\phi}}^2.
\end{equation}
So a unitary $U$ with initial state $\phi$ is called {SJK-recurrent} iff $\sum_np_n=\infty$. It is clear that the entire computation leading up to
the criterion in the Markov case, being based on conditioning, loses its significance in the quantum case. Also the operational meaning as the
probability of some property of paths of the process is lost. Therefore \cite{Stefanak} proposes an alternative operational criterion, based on the
observation that $\sum_np_n=\infty$  is equivalent to $\prod_n(1-p_n)=0$. \cite{Stefanak} defines the return probability $\RSJK$ by means
of
\begin{equation}\label{SJK}
    \RSJK=1-\prod_{n=1}^\infty(1-p_n).
\end{equation}
Operationally, this corresponds to the following procedure \cite{Stefanak}: ``Take a system and measure the position of the walker after one time
step at the origin, then discard the system. Take a second, identically prepared system and let it evolve for two time steps, measure at the origin,
then discard the system. Continue similarly for arbitrarily long evolution time. The probability that the walker is found at the origin in a single
series of such measurement records is the P\'olya number [i.e., $\RSJK$].'' This procedure is strange indeed. Imagine what it would mean in the
parable of the prodigal son. The father is clearly interested in the return of his son, but when he does not come back until year $n$, he sends out
another son and asks whether this one returns exactly in year $n+1$, ignoring all earlier and later arrivals by this son, and indeed of all those
sent before. While being rather wasteful of sons, it still does not provide an answer to the question ``Does every prodigal son eventually return?''.
Yet the procedure of \cite{Stefanak} is probably the best we can do if we want to avoid observing the quantum process. We will avoid such strangeness
in our approach but, of course, this means that we must disturb the free evolution and explicitly describe the measurements.

\subsection{Arrival by absorption}\label{sec:def}
Here we go the alternative route of explicitly introducing measurements to check the return of the system after every step.Starting at
$\phi_0=\phi$, one step thus consists of the unitary transformation $\phi_n\to U\phi_n$ followed by the measurement of the projection
$\kettbra\phi$. If this gives a positive result (for which the probability is $\abs{\braket\phi{U\phi_n}}^2$) the experiment is over. Otherwise the
system is left in the state $\phi_{n+1}=c_{n+1}(\idty-\kettbra\phi)U\phi_n$, where we choose $c_{n+1}$ so that $\norm{\phi_{n+1}}=1$. In this
iteration we will always have $\phi_n$ proportional to $\Ut^n\phi$, where
\begin{equation}\label{Ut}
    \Ut= (\idty-\kettbra\phi)U.
\end{equation}
Hence we can set $\phi_n=\Ut^n\phi/\norm{\Ut^n\phi}$. The normalization factor $s_n=\norm{\Ut^n\phi}^2$ has the direct interpretation as the {\it
survival probability}, i.e., for the walker to remain undetected for all steps including the $n^{\rm th}$. The probability for detection in step $n$,
conditioned upon survival up to $n-1$ is $\abs{\braket\phi{U\phi_{n-1}}}^2$. To get the absolute probability for first detection in step $n$ we have
to multiply this with the survival probability $\norm{\Ut^{n-1}\phi}^2$. With the normalization convention for $\phi_{n-1}$ we get the {\it first
arrival probability} in the form $\abs{a_n}^2$, where $a_n$ is the {\it first arrival amplitude}
\begin{equation}\label{an}
     a_n=\braket\phi{U\Ut^{n-1}\phi}, \quad n\geq1.
\end{equation}
The total probability for events up to and including the $n^{\rm th}$ step,
i.e., detection at step $k\leq n$ or survival, thus adds up as
$$1=\sum_{k=1}^n\abs{a_k}^2+ \norm{\Ut^n\phi}^2.$$
The {\it return probability} is therefore
\begin{equation}\label{si}
    R=\sum_{n=1}^\infty\abs{a_n}^2=1-\lim_{n\to\infty}\norm{\Ut^n\phi}^2.
\end{equation}
Accordingly, we call the pair $(U,\phi)$ {\it recurrent} if $R=1$, and {\it transient} otherwise.

This notion, which is the subject of our paper is not equivalent to SJK-recurrence, as the following example shows.

\begin{example}\label{sec:exampleDifferent}
Let us look at an elementary example, which demonstrates the difference between SJK-recurrence and the
notion we introduced. The Hilbert space will be $\HH=\Cx\oplus\ell^2(\Ir)$, in which we will denote the basis vectors as $\ket\ast$ and $\ket x,\
x\in\Ir$. The unitary is defined by $U\ket\ast=\ket\ast$ and $U\ket x=\ket{x+1}$, and the initial state will be $\phi=\alpha\ket\ast+\beta\ket0$ with
$\abs\alpha^2+\abs\beta^2=1$. Clearly, we have
\begin{equation}\label{exUn}
    \braket\phi{U^n\phi}=\abs\alpha^2(1+\delta_{n0}).
\end{equation}
To determine $\Ut_n\phi$ we can solve the recursion for the expansion coefficients, which gives
\begin{equation}\label{exUtn}
    \Ut^n\phi=\alpha\abs\beta^{2n}\,\ket\ast+\beta\,\ket n
          -\beta\abs\alpha^{2}\sum_{k=1}^{n}\abs\beta^{2(k-1)}\,\ket{n-k}
\end{equation}
Hence
\begin{equation}\label{normexUn}
    s_n=\norm{\Ut^n\phi}^2=\frac{2\abs\beta^2}{1+\abs\beta^2}+ \frac{\abs\alpha^2\abs\beta^{2n}}{1+\abs\beta^2}
\end{equation}
This converges to $s_\infty={2\abs\beta^2}/({1+\abs\beta^2})$, which is non-zero whenever $\beta\neq0$. In other words, according to our definition,
the only recurrent case in this example is $\phi=\ket\ast$, which does not move at all. All other cases are transient. In contrast,
$p_n=\abs{\braket\phi{U^n\phi}}^2=\abs\alpha^4$ for all $n>1$, so $\RSJK=1$ unless $\alpha=0$. Hence SJK-recurrence holds with the only exception of
the initial state which moves according the shift on $\ell^2(\Ir)$. \hfill$\square$
\end{example}

%% file: recurP.tex
\section{Spectral characterization of recurrence}\label{sec:proof1}
\subsection{The mathematical setting}
We have chosen our definition so that it only depends on the pair $(U,\phi)$ made up of a unitary operator on a Hilbert space $\HH$ and a unit vector
$\phi\in\HH$. It is clear from the definitions that only the subspace generated by the vectors $U^n\phi$ plays a role. We may thus ignore parts of
the Hilbert space never explored by the evolution, and assume without loss that $\phi$ is cyclic for $U$. Similarly, $(U,\phi)$ are only needed up to
unitary equivalence, i.e., the pair $(VUV^*,V\phi)$ with $V$ any unitary operator yields the same return probabilities. One invariant of pairs
$(U,\phi)$ under unitary equivalence is the $\phi$-expectation of the spectral measure of $U$. This is  the probability measure $\mu$ on the unit
circle $S^1=\{u\in\mathbb{C}:|u|=1\}$, defined by $\mu(du)=\braket\phi{E(du)\phi}$, where $E$ is the spectral measure of $U$. Equivalently, we can
characterize the system by the moments of $\mu$, i.e., by the Fourier coefficients
\begin{equation}\label{mun}
    \mu_n=\int\!\!\mu(du)\ u^n=\braket\phi{U^n\phi}, \quad n\in\Ir.
\end{equation}
The measure $\mu$ is, in fact, a complete invariant: We may choose as a canonical form of the pair $(U,\phi)$ the Hilbert space
$\HH_\mu=\Ell2(S^1,\mu)$, the unitary $U_\mu$ of multiplication by the argument $u$, and the vector $\phi_\mu(u)=1$. These are equivalent to the
original setting by virtue of the operator $V:\HH\to\HH_\mu$ defined by $(VU^n\phi)(u)=u^n$, for $n\in\Ir$, which is unitary because
$\braket{U^m\phi}{U^n\phi}=\braket{u^m}{u^n}=\mu_{n-m}$.

Various functions have been introduced to characterize probability measures on the unit circle, all of which are analytic for $\abs z<1$, and whose
boundary values for $\abs z\to1$ typically contain relevant information about the measure. We use the {\it moment generating} or {\it Stieltjes
function}
\begin{equation}\label{muh}
    \muh(z)=\sum_{n=0}^\infty\mu_nz^n=\int\frac{\mu(du)}{1-uz},
\end{equation}
the {\it Carath\'eodory function}
\begin{equation}\label{Carat}
    F(z)=\int\!\mu(du)\ \frac{u+z}{u-z}=2\,\overline{\muh}({z})-1,
\end{equation}
and the {\it Schur function}
\begin{equation}\label{Schur}
    f(z) = \frac{1}{z} \frac{F(z)-1}{F(z)+1}=\frac1z\
            \frac{\overline{\muh}({z})-1}{\overline{\muh}({z})}.
\end{equation}
Here we have used the convention that, for an analytic function $g$, the analytic function with the conjugated Taylor coefficients is denoted by
$\overline g$, i.e., $\overline{g}(z)=\overline{g(\overline z)}$.

For $\abs z<1$ we have
\begin{equation}\label{ReF}
    \RE F(z)= \int\mu(du)\,\frac{1-\abs z^2}{\abs{u-z}^2} =\frac{1-|zf(z)|^2}{|1-zf(z)|^2}>0,
\end{equation}
and hence $\abs{f(z)}<1$, due to Schwarz's lemma. The Schur function thus maps the open unit disk to itself. The proof that, conversely every analytic function on the disc
with $\RE F>0$ and $F(0)=1$ (resp. $\abs f<1$) arises in this way from a probability measure $\mu$ is the subject of classic papers by Caratheodory,
Schur,  Herglotz and many others. For a lively discussion of the history we recommend Simon's book \cite[Ch.~1.3,\ Notes]{Simon1}.

Both $f$ and $F$ have radial limits $g(e^{it})=\lim_{r\to1^-}g(re^{it})$ for almost all $t$ so that we can consider that they are extended a.e. in
the unit circle. The absolutely continuous part of $\mu$ is supported on the points $z \in S^1$ such that the radial limit satisfies $\RE F(z)>0$,
i.e., $|f(z)|<1$. In fact, the density of the absolutely continuous part is given by \eqref{ReF}. Therefore, the singular measures are characterized
by the fact that the related Schur function is ``inner", i.e. $|f(z)|=1$ a.e. in $S^1$ \cite{Simon1}. The singular part of $\mu$ is concentrated on
the points $z \in S^1$ such that $\lim_{r\to1^-}\RE F(rz)=\infty$, that is, $zf(z)=1$. Concerning the pure point part of $\mu$, the mass of any point
$z \in S^1$ is given by
\[
\mu(\{z\}) = \lim_{r\to1^-} \frac{1-r}{2} F(rz) = \lim_{r\to1^-} \frac{1-r}{2} \frac{1+rzf(rz)}{1-rzf(rz)}.
\]

\subsection{Generating functions}
Let us now come back to the return problem as posed in Sect.~\ref{sec:def}. That is, we set $\Ut=(\idty-\kettbra\phi)U$, and ask for the survival
probabilities $s_n=\norm{\Ut^n\phi}^2$ in the limit $n\to\infty$, resp.\ the sum $\sum_{n=1}^\infty\abs{a_n}^2$ over the arrival amplitudes
$a_n=\braket\phi{U\Ut^n\phi}$.

To this end we compute a generating function for the vectors $\Ut^n\phi$, starting from that of the operators $\Ut^n$,
\begin{equation}\label{geneRes}
    \widetilde G(z)=\sum_{n=0}^\infty z^n \Ut^n=(\idty-z\Ut)\inv, \quad \abs z<1.
\end{equation}
This is essentially the resolvent of $\Ut$, and can be related to the corresponding expression $G(z)=(\idty-zU)\inv$ for the unitary evolution by the
perturbation resolvent formula, which is especially efficient in this case, because the perturbation $U-\Ut$ is a rank 1 operator. For completeness
we repeat the short derivation here. We get
\begin{eqnarray}
   G(z)-\Gt(z)&=&(\idty-z\Ut)\inv\Bigl((\idty-z\Ut)-(\idty-zU)\Bigr)(\idty-zU)\inv \nonumber\\
      &=&z\,\Gt(z)\kettbra\phi UG(z)\nonumber\\
  G(z)&=&\Gt(z)\bigl(\idty+z\kettbra\phi U G(z)\bigr)\nonumber\\
 {G(z)\phi} &=&{\Gt(z)\phi}\bigl(1+z\braket\phi{ U G(z)\phi}\bigr)\label{maponphi}\\
 {\Gt(z)\phi}&=&\frac{G(z)\phi}{1+z\braket\phi{ U G(z)\phi}}.     \label{Krein1}
\end{eqnarray}
Here, at \eqref{maponphi} the previous equation was applied to $\phi$. The function in the denominator can be simplified to
\begin{equation}\label{simplifytomuh}
    1+z\braket\phi{ U G(z)\phi}=\int\!\!\mu(du)\Bigl(1+\frac{zu}{1-zu}\Bigr)=\muh(z).
\end{equation}
We now take the scalar product of \eqref{Krein1} with $U^\dag\phi$ to obtain a generating function for the arrival amplitudes, and use
\eqref{simplifytomuh} also to simplify the numerator of \eqref{Krein1}. This yields
\begin{eqnarray}
  \hata(z)&=&\sum_{n=1}^\infty a_nz^n=\sum_{n=0}^\infty\braket\phi{U\Ut^n\phi}z^{n+1} = z\,\braket\phi{U\Gt(z)\phi} \nonumber\\
          &=&\frac{\muh(z)-1}{\muh(z)}  \label{qrenewal}\\
          &=& z\,\overline{f}({z}). \label{hata}
\end{eqnarray}
That is, the Schur function is essentially the generating function for the arrival amplitudes.
This leads to an interesting analogy between the classical case and the quantum case discussed here:

\vskip12pt
\begin{center}
\begin{tabular}{r|c|c}
   &classical&quantum\\ \hline
return in $n$ steps & probability &amplitude \\
                    & $\Bigl.p_n\Bigr.$       & $\mu_n$ \\
generating function & $\widehat p$\hfil see \eqref{gfMarkov}  & $\muh$\hfil see \eqref{muh}
\\ \hline
first return after $n$ steps & probability &amplitude \\
                    & $\Bigl.q_n\Bigr.$       & $a_n$ \\
generating function & $\widehat q$\hfil see \eqref{gfMarkov}
                                     & $\hata=z\overline f$\quad\hfil see \eqref{hata}
\\ \hline
renewal equation    &$\displaystyle\widehat q=1-\frac1{\widehat p}$\quad\eqref{qhat}
                       &$\displaystyle\hata=1-\frac1{\muh}$\quad\eqref{qrenewal}
\end{tabular}
\end{center}
In the classical case one takes the return probabilities $p_n$ and uses the renewal equation to get the first return probabilities $q_n$. In our
approach the renewal equation is used at the level of amplitudes, to pass from $\mu_n$ to the arrival amplitudes $a_n$, which are then squared to
give the arrival probabilities. This is in keeping with Feynman's elementary exposition of the ``first principles of quantum
mechanics''\cite{Feynman}, losely speaking: ``In quantum theory add amplitudes  and square at the end to get probabilities''. The continuous time
limit of this idea is the root of his path integrals, and indeed there is an approach to continuous time measurements in quantum mechanics
\cite{Mensky} based on the selection of suitable subsets of paths in the path integral. On the other hand, Feynman also explains that when you
monitor the process (his example is a measurement near the double slit), you should {\it first} square and then add probabilities. This would
identify  our approach as a ``non-monitoring'' one, in an apparent contradiction to our exposition in Sect.~\ref{sec:rec}. It is debatable whether
the principle  ``Add amplitudes where the classical theory would have you add probabilities'', which works so beautifully in Feynman's Quantum
Mechanics 101, is really a fundamental law. In the end, it is Linear Algebra which can give us a unified understanding of the two columns in our
table: Both columns are based on the powers of a matrix in which the first row has been replaced by zeros: in the quantum case by multiplying with
the projection $\idty-\kettbra\phi$ and in the classical case by disregarding all jumps to 0. The renewal equation depends only on this, regardless
of further structures such as unitarity and elementwise positivity of transition operators.

We remark that the above table suggests a reading of the SJK approach as ``take the $p_n$ from quantum mechanics, and define $q_n$ by the renewal
equation''. This does lead to the SJK-recurrence criterion. However, in the paper \cite{Stefanak} the quantities $q_n$ are avoided, as described
above after \eqref{SJK}. The reason for this is not stated in the paper, but there is a good reason, which we discuss in Sect.~\ref{sec:taudoof}: The
$q_n$ can be negative.

\subsection{The recurrence criterion}\label{sect:recc_crit}
We can now utilize the identification \eqref{hata} of the generating function $\hata$ to turn the condition for recurrence, i.e.
$\sum_n\abs{a_n}^2=1$, into a condition for the Schur function $f$. To this end we fix $r<1$, and consider the series
$\hata\bigl(re^{it}\bigr)=\sum_nr^na_n\,e^{itn}$ as a Fourier series and use the Plancherel theorem:
\begin{equation}\label{ansqsum}
    \sum_nr^{2n}\abs{a_n}^2
       = \frac{r^2}{2\pi}\int_{-\pi}^\pi\!\!\!dt\ \abs{f(re^{-it})}^2.
\end{equation}
Now recall that the Schur function is bounded by $1$, so for the limit of this expression as $r\to1$ to be equal to the maximum $1$ we must have that
$\abs{f(re^{it})}\to1$ for almost all $t$. In other words, $f$ has to be an inner function. As we pointed out, this has a simple characterization
directly in terms of the measure.

\begin{theorem}\label{thm:crit}
  Consider a unitary operator $U$ and an initial vector $\phi$. Then the following are equivalent:
  \begin{itemize}
    \item[(1)] $(U,\phi)$ is recurrent.
    \item[(2)] The Schur function $f$ is inner.
    \item[(3)] $\mu$ has no absolutely continuous part.
  \end{itemize}
\end{theorem}

\begin{remark}
Theorem \ref{thm:crit} supplements known results about the interplay between the propagation properties of a state $\phi$ and its spectral measure.
For example a corollary of the Riemann-Lebesgue-lemma states that a state with absolutely continuous spectral measure will eventually leave any
finite region,  which corresponds to the statement of Theorem \ref{thm:crit} that a state can only be recurrent and thereby returning to its initial
starting point if its spectral measure contains no absolutely continuous part. In comparison to the famous RAGE theorem \cite{last,RAGE,SimonReedIII}
which characterizes the continuous and the pure point part of the spectral measure of a state by its dynamical properties our result distinguishes
between its absolutely continuous and its singular part.
\end{remark}

\subsection{A spectral view of classical and SJK-recurrence}
\subsubsection{Classical} \label{sec:classpectrum}
Here is brief account of the use of the ``spectral method" to analyze classical random walks. Not all Markov chains with state space $X$ can be
analyzed in this fashion since the matrix $P_{x,y}$ can be very far from being symmetric.

It turns out that the class of Markov chains that can be put in self-adjoint form has a nice description, namely they are the so-called {\it
reversible chains} defined as those for which there is a vector with positive components $\pi_x$ such that
\[
\pi_x P_{x,y} = P_{y,x} \pi_y.
\]
Assuming the existence of this  vector one can introduce an inner product between functions $f$ and $g$ defined on $X$ by the rule
\[
(f,g)_{\pi} = \sum_{x \in X} f_x \overline{g}_x \pi_x.
\]
One now checks easily that $(Pf,g)_{\pi} = (f,Pg)_{\pi}$ which says that the matrix $P$ is now symmetric with respect to such inner product and the
spectral theory for self-adjoint (bounded) operators can be used to great advantage.

From the spectral measure $E$ of our self-adjoint operator one gets a spectral measure $m(d\lambda)=(\psi,E(d\lambda)\psi)_\pi$, supported in
$[-1,1]$, for any probability distribution $\psi_x$. The probability of going back to the same probability distribution $\psi_x$ in $n$ steps is
given by the moments of the spectral measure, namely,
\[
m_n = \int \! m(d\lambda) \; \lambda^n = (\psi,P^n\psi)_\pi, \qquad n=0,1,2,\dots
\]

In particular, the probability $p_n$ of returning to a given initial state $0$ after $n$ steps coincides with the $n$-th moment of the spectral
measure $m$ for the distribution probability $\delta_{x,0}$ concentrated at $x=0$. This identifies the generating function of the return
probabilities $p_n$ as the Stieltjes function of the measure $m$, i.e.,
\[
\widehat{p}(z) = \sum_{n=0}^\infty p_nz^n = \int_{-1}^1 \frac{m(d\lambda)}{1-\lambda z}.
\]

We finally have that the state 0 is recurrent exactly when
\[
\int_{-1}^1 \frac {m(d\lambda)}{1-\lambda} = \infty.
\]
Besides, using the renewal equation, we find that, in the recurrent case, the expected time for the first return to this state is
\[
\tau^{C} = \sum_{n=1}^\infty nq_n = \left.\frac{d}{dz}\widehat{q}(z)\right|_{z=1} = \lim_{z\to1} \frac{1-\widehat{q}(z)}{1-z} = \lim_{z\to1}
\frac{1}{(1-z)\widehat{p}(z)} = \frac{1}{m(\{1\})}.
\]
Therefore, the process returns to a state in a finite expected time exactly when its spectral measure has a mass at $1$, and the expected return time
is the inverse of the corresponding mass.

All this shows that, classically, the recurrence properties of a state only depend on the behavior of its spectral measure around the point 1.

\subsubsection{SJK}
For the sake of comparison, we also give here a simple spectral characterization of SJK-recurrence, which is missing in \cite{Stefanak} and was
formulated in a different way in \cite{CMGVI}. Since
\begin{equation}\label{pnFourier}
    \braket\phi{U^n\phi}=\int\!\mu(du)\ u^n=\mu_n
\end{equation}
are just the Fourier coefficients of the measure, and $\mu_{-n}=\overline{\mu_n}$, the condition
\begin{equation}\label{pnEll2}
    \sum_{n=0}^\infty p_n=\frac12+\frac12\sum_{n=-\infty}^\infty \abs{\mu_n}^2 <\infty
\end{equation}
means, by the Plancherel Theorem, that $\mu(du)=\rho(u)\,du$ for a function $\rho\in\Ell2(S^1,\nu)$, where $\nu$ denotes the uniform measure on the
unit circle $S^1$. Note that since $S^1$ is compact, this is more restrictive than just absolute continuity, i.e., the existence of an integrable
density $\rho$. So we can summarize: {\it A pair $(U,\phi)$ fails to be SJK-recurrent if and only if $\mu$ is absolutely continuous with a square
integrable probability density.}

\subsection{Dependence on the starting point }
One feature, which makes P\'olya's theory of recurrence so useful is that, if the walk is irreducible, recurrence is independent of the starting
point, hence a property of the walk itself. Here ``irreducible'' means \cite{Karlin} that from any point every other can be eventually reached with
non-zero probability, a property which is usually checked quite easily.

In the quantum case the natural analogue of ``the set of points reachable from starting point $x$'' is ``the cyclic subspace generated by the initial
vector $\phi$'', i.e., the span of all $U^n\phi$, which we will briefly denote by $\cyclic\phi$. Now it is clear that in the quantum case no
non-trivial unitary operator can be ``irreducible'', because the spectral decomposition will generate many subspaces which cannot be reached from
each other. Nevertheless, some degree of independence can be stated, and follows immediately from Theorem~\ref{thm:crit}.

\begin{corollary}
  Let $\phi'\in\cyclic\phi$ and $(U,\phi)$ recurrent. Then $(U,\phi')$ is recurrent.
\end{corollary}

Note that this is not true for SJK-recurrence. As a counterexample we can take Example~\ref{sec:exampleDifferent}, or any other where $\mu$ has both
pure point and absolutely continuous components. This will be SJK-recurrent, but a suitable initial vector $\phi'$ from the absolutely continuous
subspace will be transient. A very similar argument shows that the above corollary is not true for transience. It is also not true for
SJK-transience: this fails because even if $\phi$ is absolutely continuous with square integrable density, $\cyclic\phi$ will also contain vectors
whose density is not square integrable.

We could also introduce a a notion of ``hereditary transience'' of $(U,\phi)$ meaning that all $\phi'\in\cyclic\phi$ are transient. Clearly, this is
equivalent to $\mu$ being absolutely continuous. For hereditary transience the analogue of of the above corollary holds by definition. However, for
the SJK version this is not an option (hereditary SJK-transience is vacuous). To complete the picture we might add that hereditary SJK-recurrence
means that $\mu$ has no absolutely continuous part, since that part would contain also vectors with square integrable density.  Hence hereditary
SJK-recurrence is exactly equivalent to recurrence in our sense.

\subsection{An inverse problem}
A famous result in the classical case \cite{Yosida} states that any sequence $q_n$ with $q_n\geq0$, $q_0=0$ and $\sum_nq_n=1$ can be the sequence of
first return probabilities of a Markov chain. Here we discuss the analogous questions in the quantum case. It is clear from Bochner's theorem that
the sequences $\mu_n=\braket\phi{U^n\phi}^2$ are exactly the positive definite ones with $\mu_0=1$: the kernel $k_{n,m}=\mu_{n-m}$ has to be positive
definite, where we set $\mu_{-n}=\overline{\mu_n}$ for $n>0$. Equivalently, the sequence of Toeplitz determinants \cite{Simon1} is positive.
Clearly, by expressing the $\mu_n$ in terms of the $a_n$ this implies a hierarchy of conditions on the $a_n$. The first few look encouragingly
simple:
\begin{eqnarray}
  1-\abs{a_1}^2 &\geq&0 \label{toep1}\\
  (1-\abs{a_1}^2)^2-\abs{a_2}^2&\geq&0 \label{toep2}\\
  (1-\abs{a_1}^2)^3-(1-\abs{a_1}^2)(2\abs{a_2}^2+\abs{a_3}^2)+\abs{a_2}^2&\geq& 2\RE\bigl(\overline{a_2}\,a_1a_3\bigr) \label{toep3}
\end{eqnarray}
These already suffice to answer some basic questions. First of all, we see that there is no quantum analog of the Yosida-Kakutani Theorem cited
above: there are non-trival constraints on the sequence of first return probabilities. From \eqref{toep2}, in particular, we see that the set of
admissible sequences $\abs{a_n}^2$ is not convex. Finally, the Toeplitz determinant constraints cannot be expressed in terms of the absolute values
of the $a_n$ alone. So even from the whole hierarchy of such inequalities it is not trivial to extract the precise constraints for the $\abs{a_n}$.

%% file: recurT.tex
\section{Expected return time in the recurrent case}\label{sec:time}
\subsection{It is an integer}\label{sec:integer}

Consider a recurrent pair $(U,\phi)$. Its first arrival probabilities $\abs{a_n}^2$ are then a normalized probability distribution on $\Nl$. What are
its features? Perhaps the first quantity to look at is the expected return time, i.e.,
\begin{align}\label{def:matime}
  \tau = \sum_{n=1}^\infty \abs{a_n}^2\, n.
\end{align}
It turns out that $\tau$ can be computed from the measure very easily and turns out to be always an integer or infinite. This is surprising, because
the $a_n$ clearly depend continuously on the measure $\mu$, so a small change in $\mu$ could be expected to also change $\tau$ only a little. Our aim
is therefore not just to prove this statement but to make it intuitive by giving $\tau$ a topological interpretation.

\begin{theorem}\label{thm:time}
Let $(U,\phi)$ be a recurrent pair with spectral measure $\mu$, Schur function $f$ and expected return time $\tau\in\Rl_+\cup\{\infty\}$. Then the
following are equivalent
  \begin{itemize}
    \item[(1)] $\tau<\infty$.
    \item[(2)] $f$  is a rational function.
    \item[(3)] $\mu$ is equal to the sum of $n<\infty$  distinct point measures with non-zero weights.
  \end{itemize}
Moreover, $\tau=n$, and the polynomial degree of numerator and denominator of the Schur function is $n-1$.
\end{theorem}

We will prove the implications (3)$\Rightarrow$(2)$\Rightarrow$(1) emphasizing the narrative of ideas leading to the identification $\tau=n$. The
slightly more technical implication (1)$\Rightarrow$(3) is given in a separate, slightly more formal style.

It is clear from \eqref{muh} that $\muh=P/Q$ is the quotient of two polynomials $P,Q$ of degrees $(n-1),n$, respectively. Then we can write $f$ as
$R/P$ with $R=(P-Q)/z$. $R$ is a polynomial, because $P(0)/Q(0)=1$, and hence $P-Q$ has no constant term. Clearly, its degree is also $n-1$. $R$ and
$P$ have no common zeros, since these would be common zeros of $P$ and $Q$, contradicting the observation that $\muh$ has $n$ distinct poles on the
unit circle. This proves (2).

Now observe that $f$ is a rational function which, as an inner function, must have modulus one on the unit circle, and hence the generating function
\eqref{hata} of the arrival amplitudes
\begin{equation}\label{gg}
    g(t)=e^{it}\overline f(e^{it})=\sum_{n=1}^\infty a_n e^{int}
\end{equation}
has modulus one for all real $t$. So $g(t)$ winds around the origin an integer number $w(g)$ of times as $t$ goes from $0$ to $2\pi$. Locally we can
write $g(t)=\exp(i\gamma(t))$, so the angular velocity is
\begin{equation}
   \partial_t\gamma(t)=\frac{\partial_t g(t)}{ig(t)}=\overline{g(t)}\ \frac1i\,\partial_tg(t).
\end{equation}
Integrating this over one period $t\in[0,2\pi]$, we get $2\pi w(g)$, so
\begin{equation}
   w(g) = \frac1{2\pi}\int_0^{2\pi}\mskip-15mu dt\ \overline{g(t)} \ \frac1i\,\partial_t g(t)
        = \sum_{n=0}^\infty \overline{a_n}\ (na_n)= \tau.
\end{equation}
Here we used the Plancherel identity and the observation from \eqref{gg} that the Fourier coefficients of $(\partial_tg)/i$ are $na_n$. Hence we have
a topological interpretation of the expected return time $\tau$, which already makes clear that it must be an integer. We remark that this connection
was found in a totally different context by H.~Brezis following a hint by I.~M.~Gelfand \cite{Brezis}.

To determine this integer we use the Blaschke factorization, which decomposes a rational inner function into a product of the form
\begin{equation}\label{BlaschkeRat}
    f(z)=\beta\prod_{k=1}^{n-1}\frac{\alpha_k-z}{1-\overline{\alpha_k}\,z}
\end{equation}
with $\abs{\alpha_k}<1$ and $\abs\beta=1$. The $\alpha_k$ are exactly the zeros of $f$ inside the unit circle, which by the reflection principle must
be matched by poles at $1/\overline{\alpha_k}$ outside the circle. The upper limit of the product expresses our previous analysis of the degrees of
the rational function $f$. Now the winding number of a product of several functions is clearly the sum of the winding numbers. Each of the Blaschke
factors contributes $1$. Indeed, the winding number of a single factor clearly depends continuously on $\alpha$, so must be constant, and we may set
$\alpha=0$ (i.e., $f(z)=z$) to evaluate it.  Observing that the generating function \eqref{hata} contains an additional factor $z$ we get
$\tau=w(g)=n$, proving (1).

Before proving the remaining implication, let us see how a small additional point mass changes the Schur function and the arrival time distribution.
An example is shown in Fig.~\ref{fig:phasetwist}. It is based on the sum of two point masses at $u=1$ and $u=i$. We add a small point mass at $u=-1$
by convex combination with a small weight $\varepsilon$. This leads to an additional zero of $f$ close to $z=-1$, and a full turn of the phase near
that point. The interval in which this turn happens is of order $\varepsilon$. Thus for fixed low $n$ the contribution to the Fourier integral for
$a_n$ is negligible. For higher $n$ it becomes visible, and the unit contribution to the expected return time comes from the large $n$ Fourier
coefficients. We will discuss this in more detail in Sect.~\ref{sec:finite}.

\begin{figure}[ht]
\centering
 \includegraphics[width=12cm]{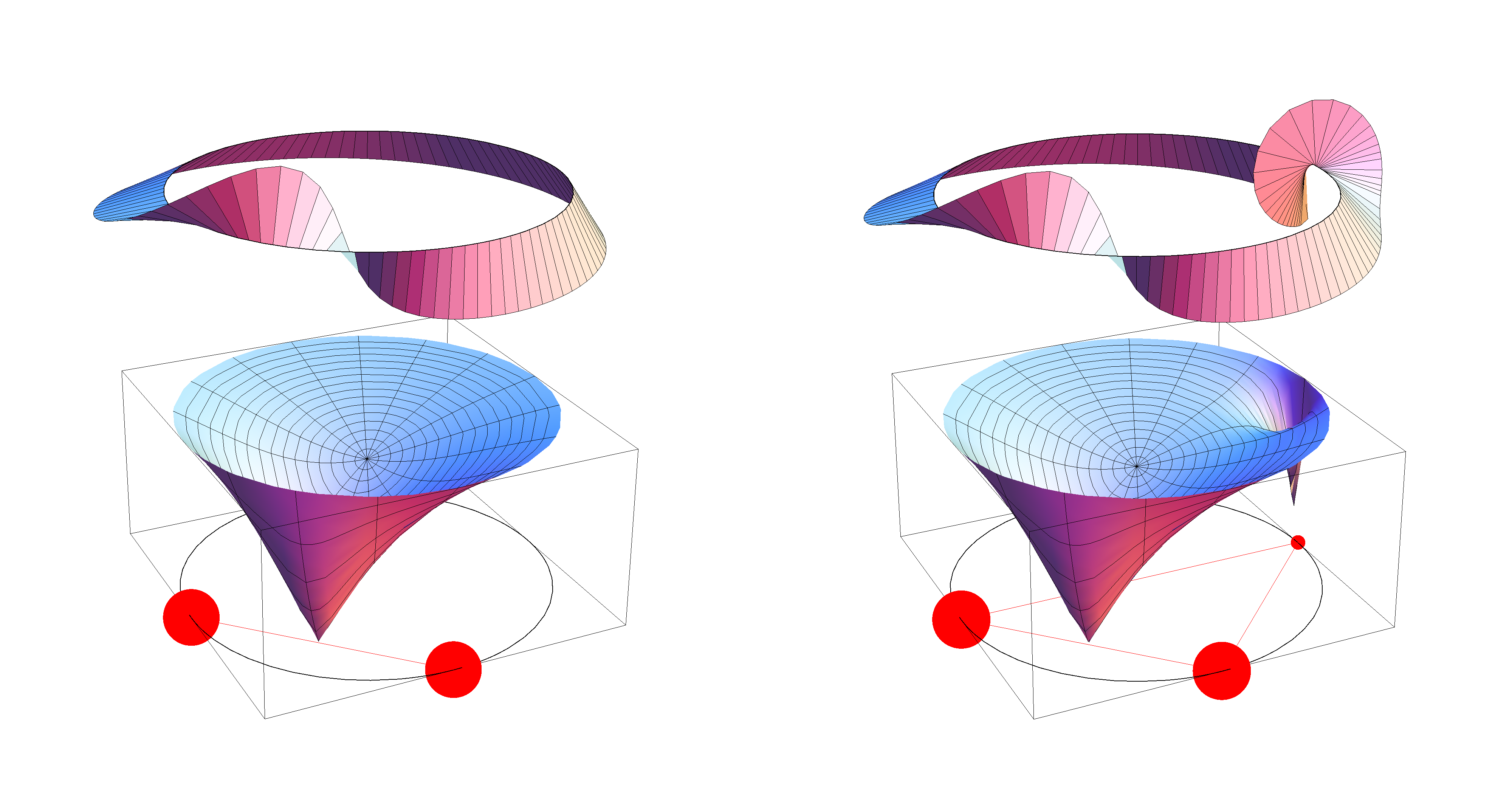}
  \caption{Schur function for a sum of two and three point measures. Bottom: unit circle with positions of the masses, represented according to weight.
  Middle: absolute value of $f$. Top: the phase of $f$ on the unit circle. The ribbon is traced out by a vector indicating the argument of $f$.
  It has fixed length, is based on a unit circle (thick black) and orthogonal to it.
  Left panel: two point masses at $1$, $i$ with equal weight. Right: same mixed with a point mass near $-1$ with weight $0.01$.
  The narrow spike near the light point mass does touch the zero plane: it is a zero of $f$ near the circle with large derivative. }
  \label{fig:phasetwist}
\end{figure}

\begin{proof}[of Theorem~\ref{thm:time}, part (1)$\Rightarrow$(3)]
We prove this by contradiction, i.e., we show that for non-rational inner functions $f$ we have $\tau=\infty$. We combine the expression
\eqref{ansqsum} with the observation that
\begin{equation}\label{derivr2}
    \frac{1-r^{2n}}{1-r^2}=\sum_{k=0}^{n-1} r^{2k}
\end{equation}
increases monotonely to $n$ as $r\to1^-$. Hence,
\begin{eqnarray}
    \tau(f)=\lim_{r\to1^-}\sum_{n=0}^\infty\frac{1-r^{2n}}{1-r^2}\abs{a_n}^2
          &=&\lim_{r\to1^-}\frac{1}{2\pi}\int_{-\pi}^{\pi}\!\!dt\ \frac{1-r^2\abs{f(re^{it})}^2}{1-r^2} \label{taulimit} \\
          &=&1+\lim_{r\to1^-}\frac{1}{2\pi}\int_{-\pi}^{\pi}\!\!dt\ \frac{1-\abs{f(re^{it})}^2}{1-r^2}. \label{taulimit2}
\end{eqnarray}
We now consider a product of inner functions $f_1,f_2$, which is, of course, again an inner function.
Then $\abs{f_1f_2}\leq\abs{f_1}$ and, consequently
\begin{equation}
    \tau(f_1f_2)\geq\tau(f_1).
\end{equation}
Suppose now that $f$ can be represented as a convergent Blaschke product
\begin{equation}\label{BlaschkeInf}
    f(z)=\beta\prod_{k=1}^{\infty}\frac{\overline{\alpha_k}}{\abs{\alpha_k}}\ \frac{\alpha_k-z}{1-\overline{\alpha_k}\,z},
\end{equation}
where the phase factors in the product have been introduced to make each factor positive at $z=0$ thus avoiding erratic behavior of conventional
phases. Clearly, we can take out the first $N-1$ factors as $f_1$, for which we know $\tau(f_1)=N$. Hence $\tau(f)\geq N$ for all $N$, so
$\tau(f)=\infty$.

Not every inner function is a Blaschke product, but we can invoke a result of Frostman \cite{Frostman} (see also \cite[Thm~6.4.]{Garnett}). It states
that, given an inner function $f(z)$, its M\"{o}bius transform
\begin{equation}\label{Moeboff}
    f_\alpha(z) := \frac{\alpha-f(z)}{1-\overline{\alpha}f(z)}
\end{equation}
is a Blaschke product for some $\alpha$ in the unit disk. Indeed, Frostman proved that this statement holds for almost all $\alpha$ (outside a set of
null capacity), but one $\alpha$ will be enough for us. Since $f(z)=({\alpha-f_\alpha(z)})/({1-\overline{\alpha}f_\alpha(z)})$, the two functions
satisfy the estimate
\begin{equation}
  1-|f(z)|^2
    = \frac{1-|\alpha|^2}{|1-\overline{\alpha}f_\alpha(z)|^2}\,(1-|f_\alpha(z)|^2)
    \geq \frac{1-|\alpha|}{1+|\alpha|}\,(1-|f_\alpha(z)|^2).
\end{equation}
Hence from \eqref{taulimit2},
\begin{equation}\label{deltatauMoeb}
    \tau(f) \geq 1 + \frac{1-|\alpha|}{1+|\alpha|}\, (\tau(f_\alpha)-1).
\end{equation}
Now $f_\alpha$ is a Blaschke product. If it is finite, then $f$ is rational and hence itself a finite Blaschke product, a case already covered. When
it is an infinite product, then $\tau(f_\alpha)=\infty$ and hence $\tau(f)=\infty$.
\end{proof}

\subsection{Expected arrival in the SJK scheme}\label{sec:taudoof}
Although absent in SJK papers, a definition of expected return time can be given for the SJK approach too. Let us begin with a seemingly natural
approach which, however, turns out not to work. Recall P\'olya's Theory from Sect.~\ref{sec:rec}: it allowed the determination of the probabilities
$q_n$ for first return in $n$ steps from the probabilities $p_n$ to be back at time $n$. Since the SJK-approach is based on applying P\'olya's
criterion to the $p_n$ as computed from quantum mechanics, we already have a candidate for $q_n$ and hence the expected return time via the renewal
equation \eqref{qhat}. We get (compare also Sect.~\ref{sec:classpectrum})
\begin{equation}\label{taudoof}
    \taudoof
    =\left.\frac{d\widehat q(z)}{dz}\right|_{z=1}=\lim_{z\to1}\frac{1-\widehat q(z)}{1-z}
    =\left(\lim_{z\to1}(1-z)\widehat p(z)\right)^{-1}.
\end{equation}
If the measure consists of finite point masses the $\mu_n$ are given by $\sum_k m_k \exp(i\theta_k n)$ and we can further evaluate \eqref{taudoof} to
\begin{align}\label{taudoof_2}
  \taudoof = \left(\lim_{z\to1}\sum_{k,l} m_k m_l \frac{1-z}{1-z e^{i(\theta_k-\theta_l)}}\right)^{-1}=\left(\sum_k m_k^2\right)^{-1}.
\end{align}
We hasten too add that in the quantum case $\taudoof$ has no operational meaning at all. What is worse, the ``first return probabilities'' $q_n$ on
which it is based need not be probabilities at all (see Example \ref{ex:Rotation}).

Therefore, in the spirit of \cite{Stefanak}, we build a notion of expected return time on the process described after \eqref{SJK}: It is a Markov
process with state space $0,1,\ldots$, such that from a state $n>0$ one goes to $0$ with probability $p_{n}$, and to $n+1$ otherwise. The probability
for first return to $0$ in $n$ steps and the corresponding expected return time are thus
\begin{eqnarray}\label{qSJK}
   \qSJK_n&=&p_n\ \prod_{k=1}^{n-1}(1-p_k), \nonumber\\
   \tauSJK&=& \sum_n n\,\qSJK_n,
\end{eqnarray}
with $\qSJK_0=0$ and $\qSJK_1=p_1$. As discussed above, this adds up to $\sum_n\qSJK_n=\RSJK$. For a given pair $(U,\phi)$ we thus get three sets of
``first return probabilities'': the $\abs{a_n}^2$ with the first return amplitudes $a_n$ from \eqref{an}, and two sets based on the probabilities
$p_n=\abs{\braket\phi{U^n\phi}}^2$, namely $\qSJK_n$ after \eqref{qSJK} and $q_n$ after Sect.~\ref{sec:classNotion}.  The following example shows
that all three notions can actually lead to different values for the same dynamical system.

\begin{example}\label{ex:Rotation}
In the Hilbert space $\Cx^2$ consider as  the unitary operator $U$ the real rotation matrix by some angle $\theta$ and choose the vector $\phi=(1,0)$
as the initial state. The corresponding spectral measure consists of two mass points of equal weights $1/2$ at $\exp(\pm i \theta)$, which
corresponds to the Schur function
  \begin{align}
    f(z) = \frac{\cos(\theta)-z}{1-\cos(\theta) z}\; .
  \end{align}
Given the Schur function, we can calculate the three series of first return probabilities $\abs{a_n}^2$, $\qSJK_n$ and $q_n$. In Fig.~\ref{fig:qqq} these three are
shown together for $\theta=\pi/4$. In that particular case $U^2\phi=(0,1)$, so $\mu_2=\braket\phi{U^2\phi}=0$, and $p_n=1,1/2,0,1/2,\ldots$ repeated
ad infinitum. Now this cannot be the return probability sequence of any Markov chain: If there is a non-zero probability to stay at site $0$ for one
step, this probability must be non-zero for all $n$. So $p_2=0$ is impossible classically. In the jargon of quantum information theory the unitary we
have chosen is called a square root of the {\tt NOT} gate \cite{sqnot}, since $U^2$ swaps the two bit values, a feat which cannot be realized by two
equal classical probabilistic steps. Indeed, $q_2=-1/4$, and the sum of all the negative $q_n$ is around $-0.77$.
  \begin{figure}[ht]
\centering
 \includegraphics[width=0.8\textwidth]{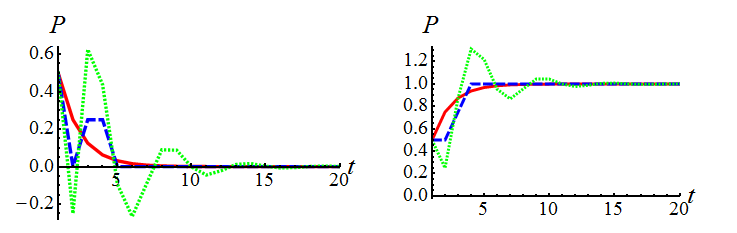}
  \caption{The figure shows the first return probabilities (left panel) and the accumulated first return probabilities (right panel) for $\theta=\pi/4$.
  The three curves correspond to the first return probabilities given by $\abs{a_n}^2$ (thick), $\qSJK_n$ (long dashed) and  $q_n$ (short dashed).}
  \label{fig:qqq}
\end{figure}
It is also clear from Fig.~\ref{fig:qqq} that there is no general inequality connecting the first returns in the SJK approach and in our approach:
which of the two gives higher probability for the event ``arrival before $t$'' may depend on $t$.

\begin{figure}[ht]
\centering
 \includegraphics[width=0.5\textwidth]{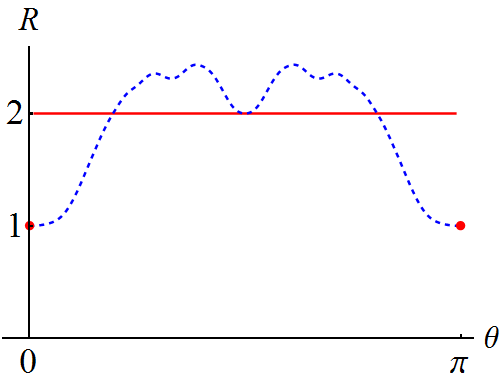}
  \caption{Dependence of the expected return times  $\tau=\taudoof$ (thick) and $\tauSJK$ (dashed, first 400 summands) on the rotation angle $\theta$ (for $\theta \in\left[-\pi,0\right]$ the graph shows the same behavior).}
  \label{expectedreturn}
\end{figure}

To see that this is also true for the expected return times we evaluate the expressions for $\tau$, $\taudoof$ and $\tauSJK$ for this example. Since
the overlap of $\phi=(1,0)$ with the two eigenvectors of $U$ is independent of the rotation angle $\theta$, Theorem \ref{thm:time} tells us that
$\tau$ equals the number of distinct eigenvalues of $U$. This amounts to $\tau=2$ for all $\theta$ except the singular cases $0$ and $\pm\pi$, where
$U$ has two equal eigenvalues and therefore $\tau$ is $1$.  Since both mass points have equal weights $1/2$ equation \eqref{taudoof_2} for $\taudoof$
gives the same result.

To determine $\tauSJK$ we have to evaluate \eqref{qSJK} which, since $p_n=\cos^2(n\theta)$ holds, amounts to
  \begin{align*}
  \tauSJK = 1 + \sum_{n\geq1}\prod_{k=1}^n\sin^2(k\theta).
  \end{align*}

Fig.~\ref{expectedreturn} shows the two curves of the three expected return times $\tau$, $\taudoof$ and $\tauSJK$ depending on the rotation angle.
By varying $\theta$, we find that $\tau<\tauSJK$ as well as $\tauSJK<\tau$ can be realized. In addition by choosing an initial state $\phi$ that does
not have equal overlap with both eigenvectors also $\taudoof<\tauSJK$ can be realized. \hfill$\square$
\end{example}

We have seen that a single point measure contribution in $\mu$ makes a process SJK-recurrent. The following proposition makes an even stronger
statement.

\begin{proposition}
Let $\mu$ be a measure on the unit circle $S^1$ with at least one atom. Then $\tauSJK<\infty$.
\end{proposition}

\begin{proof}
Let $s_n=\prod_{k=1}^n(1-p_k)$ be the survival probabilities in the SJK-approach. Then $q_n+s_{n}=s_{n-1}$ and, by partial summation,
\begin{equation}\nonumber
    \sum_{n=1}^mnq_n=\sum_{n=0}^{m-1}(n+1)s_n-\sum_{n=1}^mns_n
      =\sum_{n=0}^{m-1}s_n -ms_m\leq\sum_{n=0}^{m-1}s_n.
\end{equation}
Using the inequality $e^x \geq 1+x$, $x\in\Rl$, we find
\begin{equation}\nonumber
    \tauSJK \leq 1 + \sum_{n\geq1} \prod_{k=1}^n(1-|\mu_k|^2)
        \leq 1 + \sum_{n\geq1} e^{-\sum_{k=1}^n|\mu_k|^2}.
\end{equation}
Wiener's theorem on the unit circle states that \cite{Simon2}
\[
\lim_{n\to\infty} \frac{\sum_{k=1}^n|\mu_k|^2}{n} = \lim_{n\to\infty} \frac{\sum_{k=-n}^n|\mu_k|^2}{2n+1} = \mu_{pp}(S^1),
\]
where $\mu_{pp}$ is the pure point part of the spectral measure. Hence, the presence of mass points implies that
\[
\frac{\sum_{k=1}^n|\mu_k|^2}{n} \geq M > 0, \qquad n \geq 1,
\]
for some constant $M$. As a consequence,
\[
\tau^{SJK} \leq 1 + \sum_{n\geq1} e^{-Mn} < \infty.
\]
\end{proof}

%% file: recurF.tex
\section{The finite case}\label{sec:finite}
In this section we consider a finite dimensional system, so the measure $\mu$ is a sum of point measures at positions $u_i$ with weight $m_i$,
$i=1,\ldots,n$. Thus
\begin{equation}\label{muhfin}
   \muh(z)=\sum_{i=1}^n\frac{m_i}{1-u_iz},
\end{equation}
and the Schur function $f$ is a finite Blaschke product \eqref{BlaschkeRat}. Our plan in this section is to give a more detailed account of the
changes in the return probability when $n$ changes and hence the expected return time $\tau=n$ jumps. In particular, we will show that at these
points the variance of the first return diverges. We begin with a formula for the variance in terms of the zeros $z_j$ of $f$.

\begin{proposition}\label{variancefin}
Let $f(z)=\sum_k\overline{a_k}z^{k-1}$ be a rational Schur function with zeros $z_1,\ldots,z_{n-1}$. Then
\begin{equation}\label{var}
    \vartau=\sum_kk^2\abs{a_k}^2-\Bigl(\sum_kk\abs{a_k}^2\Bigr)^2=\sum_{j,\ell}\frac{2\,z_\ell\,\overline{z_j}}{1-z_\ell\,\overline{z_j}}.
\end{equation}
\end{proposition}

\begin{proof}
We took the connection between $f$ and the $a_n$ from our generating function relation \eqref{hata}. For the sake of this proof, however, it is
easier to omit the complex conjugate, which is anyhow irrelevant in \eqref{var}, and to shift the sequence $a_n$ so that the $a_n$ are really the
Taylor coefficients of $f$. Note that $\sum_kk\abs{a_k}^2=n-1$ is then the number of zeros.

From the Blaschke product \eqref{BlaschkeRat} we get
\begin{equation}\label{fprimeoverf}
    \frac{f'}f=\sum_j\Bigl(\frac{-1}{z_j-z}-\frac{-\overline{z_j}}{1-\overline{z_j}\,z}\Bigr).
\end{equation}
Since $\abs f=1$ on the unit circle, we can express the second moment as
\begin{eqnarray}
    \sum_kk^2\abs{a_k}^2
      &=&\int_0^{2\pi}\!\!\frac{d\theta}{2\pi}\ \left|\frac{df(e^{i\theta})}{d\theta}\right|^2
       = \int_0^{2\pi}\!\!\frac{d\theta}{2\pi}\ \left|ie^{i\theta}\frac{f'}{f}\right|^2\nonumber\\
      &=&\oint\frac{dz}{2\pi i z}  \left|\sum_j\frac{1}{z-z_j}+\frac{\overline{z_j}}{1-\overline{z_j}\,z}\right|^2,
\end{eqnarray}
where the circle indicates complex integration around the unit circle. Writing the absolute value as $\overline X\,X$ we combine the factor
$\overline X$ with the factor $1/z$. Then, using $\abs z^2=1$, all denominators are linear in $z$. This gives
\begin{equation}
    \sum_kk^2\abs{a_k}^2=\sum_{j,\ell}\oint\frac{dz}{2\pi i}\Bigl(\frac{1}{1-\overline{z_\ell}\,z}+\frac{z_\ell}{z-z_\ell}\Bigr)
    \Bigl(\frac{1}{z-z_j}+\frac{\overline{z_j}}{1-\overline{z_j}\,z}\Bigr).
\end{equation}
The integral can be solved by residue inspection. The first term in the first parenthesis and the second term in the second parenthesis are analytic
in the disc. The product of the other two terms gives a zero integral when $z_j=z_\ell$ because the pole is then of second order, and also gives zero
otherwise because it is proportional to $1/(z-z_j)-1/(z-z_\ell)$. Combining the remaining terms gives
\begin{equation}
    \sum_kk^2\abs{a_k}^2
       =\sum_{j,\ell}\Bigl(\frac{1}{1-\overline{z_\ell}z_j}+\frac{z_\ell\overline{z_j}}{1-\overline{z_j}z_\ell}\Bigr)
       =\sum_{j,\ell}\frac{1+z_\ell\overline{z_j}}{1-z_\ell\overline{z_j}}.
\end{equation}
Now in order to get the variance, we have to subtract $(n-1)^2$ from this expression, or $1$ from each term. This gives the expression stated in the
proposition.
\end{proof}

The expression \eqref{var} is positive, because we can write it as
\begin{equation}\label{var1}
    \vartau=2\sum_{s=1}^\infty\left|\sum_{k=1}^{n-1}z_k^s\right|^2.
\end{equation}
This can vanish only if $\sum_kz_k^s=0$ for all $s$. But since by Newton's relations \cite{Macdonald} the coefficients of the polynomial
$\prod_k(z-z_k)$, i.e., the elementary symmetric functions of $(z_1,\ldots,z_{n-1})$, can be expressed in terms of the sums of powers, the equation
for the $z_k$ can only be $z^{n-1}=0$. In other words, $\vartau=0$ implies that $f(z)=\beta z^{n-1}$, $|\beta|=1$, corresponding to $n$ equal point
masses at the $n$-th roots of $\overline\beta\in S^1$. This should not surprise us: It corresponds precisely to a ``clock'' unitary, which cyclically
rotates $n$ basis states, of which one is the initial state. Clearly, this returns to the origin after exactly $n$ steps. For fixed $n$, small
variance still implies that all $z_k$ are small.

In the opposite direction, suppose that one of the zeros, say $z_1$, approaches the unit circle, while the others are fixed. Then the term with
$\ell=j=1$ in \eqref{var} diverges, while all other terms stay finite. Hence $\vartau\to\infty$. Characteristically this happens when the number $n$
of point measures, and hence the number of zeros changes its value. Indeed this could happen either by {\it fusion}, when two points converge to each
other (e.g., $u_{n}\to u_{n-1}$), or by a {\it shrinking}, when the weight of one point goes to zero (e.g., $m_n\to0$). In either case we get a
weakly convergent sequence of measures for which $n$ changes in the limit. But by weak convergence the sequence of Schur functions converges
uniformly on every disk of radius $(1-\varepsilon)$. This implies that the number of zeros of the Schur function in any such disc converges. In other
words, the change in the number of zeros can only happen by one zero converging to the unit circle, which entails $\vartau\to\infty$, as we have
seen. When the change in the measure is localized near one point, as for fusion and shrinking, then the uniform convergence argument also holds for a
neighborhood of other parts of the unit circle, and we can even say that the vanishing zero has to converge to the point with shrinking weight (as
shown in Fig.~\ref{fig:phasetwist}) or to the fused point.

One source of intuition for such connections between the point measures and the Schur zeros is the following electrostatic interpretation. Although
it is stated for finite sums of mass points, which is the case under study in this section, it can be easily generalized to infinitely supported
measures too.

\begin{lemma}
Consider a measure $\mu$ on the unit circle, which is a finite sum of point measures, so $\muh(z)=\sum_{i=1}^n{m_i}/({1-u_iz})$ with $m_i>0$,
$\sum_im_i=1$, and $\abs{u_i}=1$. Then the zeros $z_1,\ldots,z_{n-1}$ of the Schur function are precisely the stationary points of the
two-dimensional electrostatic potential
\begin{equation}\label{electricV}
    V(z)=\sum_{i=1}^n m_i\log\Bigl|{u_i-z}\Bigr|.
\end{equation}
In particular, all $z_k$ lie in the convex hull of the $u_i$.
\end{lemma}

The $m_i$ are the point charges of this electrostatic system located at the $u_i$. So each of the point masses generates a repulsive force field. If
one of the charges $m_i$ is small, there will be a stationary point close to this charge at $u_i$, as seen in Fig.~\ref{fig:phasetwist}. The convex
hull of the $u_i$ is indicated in red in Fig.~\ref{fig:phasetwist}.

\begin{proof}
Clearly, by \eqref{Schur} the zeros of $f$ are those of
\begin{equation}\nonumber
    \frac1z\Bigl(\overline{\muh}(z)-1\Bigr)
       =\frac1z\sum_im_i\left(\frac1{1-\overline{u_i}z}-1\right)
       =\sum_i\frac{m_i\overline{u_i}}{1-\overline{u_i}z}
       =\sum_i\frac{m_i}{u_i-z}.
\end{equation}
It is easily verified that this is equal to
\begin{equation}\nonumber
    \frac{\partial V(x+iy)}{\partial x}+i\frac{\partial V(x+iy)}{\partial y}
\end{equation}
so at the zeros of $f$ the gradient of $V$ vanishes. The statement about the convex hull may seem obvious from this. More formally, suppose to the
contrary that the above expression vanishes for some $z$ outside the convex hull of the $u_i$. Then there must be a linear functional
$z\mapsto\RE(\lambda z)$ and a constant $c$ such that $\RE(\lambda u_i)\leq c$ for all $i$, but $\RE(\lambda z)>c$. But then
\begin{equation}\nonumber
    0=\RE\lambda\sum_i\frac{m_i}{\overline{u_i}-\overline{z}}
     =\sum_i\frac{m_i}{\abs{u_i-z}^2}\ \RE(\lambda(u_i-z))
     \leq \sum_i\frac{m_i(c-\RE(\lambda z))}{\abs{u_i-z}^2} <0,
\end{equation}
which is a contradiction.
\end{proof}

We note that the arguments which we gave for the connection between low variance of first return times and concentration of Schur zeros near the
origin is valid only for fixed $n$. i.e., explicit bounds, which could be derived from these arguments will depend on $n$. This is shown by the
following example, in which all zeros (an increasing number) go the unit circle, and the variance goes to zero. We choose the zeros to be
$z_j=\lambda e^{2\pi ik/(n-1)}$, for $k=1,\ldots,n-1$. Then in the sum \eqref{var1} the sum in the absolute value is non-zero only when $s$ is a
multiple of $(n-1)$, in which case it is equal to $(n-1)\lambda^s$.  Summing over these multiples and noting that $s=0$ is excluded also, we find
\begin{equation}\label{vartauex}
    \vartau=\frac{(n-1)^2\ \abs\lambda^{2(n-1)}}{1-\abs\lambda^{2(n-1)}}.
\end{equation}
Hence if we choose $\log\abs\lambda^2=-3(n-1)^{-1}\log(n-1)$ we have $\abs\lambda\to1$ and $\vartau\asymp (n-1)^{-1}\to0$.

%% file: recurA.tex
\section{Examples with absolutely continuous measures} \label{sec:as}

In this section we deal with transient pairs $(U,\phi)$, i.e., those whose spectral measure $\mu$ has an absolutely continuous part. We focus on ways
to get the first return probabilities $\abs{a_n}^2$ or at least their sum, the total return probability
\[
    R=\sum_{n=1}^\infty |a_n|^2= \int_{-\pi}^\pi \frac{dt}{2\pi}\,|f(e^{it})|^2 =\norm f^2.
\]
A typical class of such models are so-called quantum walks, i.e., systems with a discrete spatial degree of freedom, which in each step can change
only by a finite distance. There is a new interest in quantum walks coming from recent experimental realizations \cite{Karski:2009}. These systems
were also the motivation for \cite{Stefanak}, as well as for us. This section is therefore split between two approaches to quantum walks, which have
their characteristic strengths, but also a large overlap.

The first method, introduced by two of us (F.A.G.\ and L.V.) with coauthors \cite{CMGVI,CGMV2010} builds on a method turning {\it any} pair
$(U,\phi)$ into a quantum walk. This walk most naturally lives on $\mathbb{Z}_+=\{0,1,2,\dots\}$, although it can be extended to deal with
$\mathbb{Z}$. It almost always has position dependent steps. When the dynamics is periodic from some site onwards, one can determine the Schur
function by a fixed point equation, which typically gives absolutely continuous spectrum plus some discrete eigenvalues related to the initial
segment.

The second method is designed for strictly translation invariant systems in any lattice dimension. It has been introduced in \cite{Scudo} and was
used by two of us (A.H.W.\ and R.F.W.) with coauthors \cite{timeRandom} for getting the asymptotic position distribution even in the presence of time
dependent noise. The core of the method is the joint diagonalization of $U$ and the translations, which amounts to diagonalizing a momentum dependent
matrix. Finite range spatial variations in the dynamics can be dealt with via perturbation theory \cite{molecules}, which again produces additional
discrete eigenvalues analogous to the introduction of impurities into a periodic potential.

Both methods apply to so-called ``coined walks'' in one dimension, which live on the Hilbert space $\ell^2(\Ir)\otimes\Cx^2$, whose basis vectors we
write as $\ket{x,\alpha}$, $x\in\Ir$, and $\alpha\in\{\uparrow,\downarrow\}$. The time step unitary is factorized into ``shift'' and ``coin'' $U=SC$
with $S\ket{x,\uparrow}=\ket{x+1,\uparrow}$ and $S\ket{x,\downarrow}=\ket{x-1,\downarrow}$, and a coin operation, which acts at each site $x$
separately as a $2\times2$ unitary matrix $C_x$ on $\ket{x,\updownarrow}$. The translation invariant case corresponds to a site independent coin
$C_x\equiv C_0$, and will be treated in both approaches. In contrast to  the classical result of P\'{o}lya, who found recurrence for one dimensional
unbiased nearest neighbor random walks the coined quantum walks are always transient. Intuitively, this may be related to the much slower spreading
($\sim\sqrt t$ rather than $\sim t$) of the classical walks: they just spend more time close to home.

\subsection{Return probability from orthogonal polynomials}
How can we turn any pair $(U,\phi)$ into a quantum walk? The basic idea is to build a basis of the Hilbert space from the Gram-Schmidt
orthogonalization of the sequence $\phi$, $U^{-1}\phi$, $U\phi$, $U^{-2}\phi$, $U^2\phi,\ldots$, which is equivalent to building the $\mu$-orthogonal
Laurent polynomials on the unit circle. This leads to a canonical five-diagonal matrix representation of $U$, a  so-called CMV matrix
\cite{CMV,Simon1,SimonCMV,Watkins}
\[
\mathcal{C} = \text{\small $
\begin{pmatrix}
\overline{\gamma}_0 & \rho_0 &  0 &  0 &  0 &  0 & \dots
\\
\rho_0\overline\gamma_1 &  -\gamma_0\overline{\gamma}_1 &  \rho_1\overline\gamma_2 &  \rho_1\rho_2 &  0 &  0 & \dots
\\
\rho_0\rho_1 &  -\gamma_0\rho_1 &  -\gamma_1\overline{\gamma}_2 &  -\gamma_1\rho_2 &  0 &  0 & \dots
\\
0 &  0 &  \rho_2\overline\gamma_3 &  -\gamma_2\overline{\gamma}_3 &  \rho_3\overline\gamma_4 &  \rho_3\rho_4 & \dots
\\
0 &  0 &  \rho_2\rho_3 &  -\gamma_2\rho_3 &  -\gamma_3\overline{\gamma}_4 &  -\gamma_3\rho_4 & \dots
\\
0 &  0 &  0 &  0 &  \rho_4\overline\gamma_5 &  -\gamma_4\overline{\gamma}_5 & \dots
\\
0 &  0 &  0 &  0 &  \rho_4\rho_5 &  -\gamma_4\rho_5 & \dots
\\
\dots &  \dots &  \dots &  \dots &  \dots &  \dots & \dots
\end{pmatrix}$},
\qquad
\begin{aligned}
& |\gamma_k|<1, \\ & \rho_k=\sqrt{1-|\gamma_k|^2}.
\end{aligned}
\]
Then the pair $(\mathcal{C},e_0)$, with $e_0=(1,0,0,\dots)^t$ the first canonical basis vector is unitarily equivalent to $(U,\phi)$.

The coefficients $\gamma_k$ are called the Schur parameters of $f$ or the Verblunsky coefficients of the measure $\mu$ \cite{Verblunsky,Simon1}. They
yield a continued fraction expansion of the Schur function $f$ through the so called Schur algorithm \cite{Schur1918}
\begin{equation}\label{SchurIterate}
    f_0(z)=f; \qquad f_{k+1}(z)=\frac{1}{z}\frac{f_k(z)-\gamma_k}{1-\overline{\gamma_k}f_k(z)}, \qquad \gamma_k=f_k(0).
\end{equation}
Each $f_k$ is a Schur function in its own right, which has Schur parameters $\gamma_k$, $\gamma_{k+1}$, $\gamma_{k+2},\dots$. According to a
remarkable result due to S.~Khrushchev \cite{Khrushchev} it is closely related to the pair $(\mathcal{C},e_k)$ starting at the basis vector $e_k$,
which has the Schur function $f_kB_k$, where $B_k$ is a finite Blaschke product. Hence $\norm{f_k}^2$ is equal to the total return probability
starting from $e_k$.

Coined walks on $\Ir_+$ are a special case of this general construction, namely the case that the odd Schur parameters $\gamma_{2x+1}$ vanish. Thus
we get the CMV matrix
\begin{equation}\label{eq:walk_cmv}
\mathcal{C} = \text{\small $
\begin{pmatrix}
\overline{\gamma}_0 &  \rho_0 &  0 &  0 &  0 &  0 & \dots
\\
0 &  0 &  \overline\gamma_2 &  \rho_2 &  0 &  0 & \dots
\\
\rho_0 &  -\gamma_0 &  0 &  0 &  0 &  0 & \dots
\\
0 &  0 &  0 &  0 &  \overline\gamma_4 &  \rho_4 & \dots
\\
0 &  0 &  \rho_2 &  -\gamma_2 &  0 &  0 & \dots
\\
0 &  0 &  0 &  0 &  0 &  0 & \dots
\\
0 &  0 &  0 &  0 &  \rho_4 &  -\gamma_4 & \dots
\\
\dots &  \dots &  \dots &  \dots &  \dots &  \dots & \dots
\end{pmatrix}$}.
\end{equation}
Indeed, if we relabel our basis as $e_{2x}=\ket{x,\uparrow}$ and $e_{2x+1}=\ket{x,\downarrow}$ we get precisely the coined quantum walk on the half
line with reflecting boundary condition and the coins
\begin{equation}\label{CMVcoin}
    C_x = \begin{pmatrix} \rho_{2x} & -\gamma_{2x} \\ \overline\gamma_{2x} & \rho_{2x} \end{pmatrix},
    \qquad |\gamma_{2x}|\leq1,
    \qquad \rho_{2x}=\sqrt{1-|\gamma_{2x}|^2}.
\end{equation}
Conversely, given a sequence of coins, we can immediately build the CMV representation $\mathcal{C}$ of the corresponding walk. We will assume from
now on that $|\gamma_{2x}|<1$, which makes the walk ``irreducible'', i.e., there is no finite interval in $x$ which is decoupled from its complement.

Because the odd Schur parameter vanish the iteration \eqref{SchurIterate} gives $f_{2x}(z)=z^{-1}f_{2x-1}(z)$. We conclude that the states
$\ket{x-1,\downarrow}$ and $\ket{x,\uparrow}$ of an irreducible two-state quantum walk on $\mathbb{Z}_+$ have the same return probability, and this
return probability does not depend on the coins $C_y$, $y<x$. This is a remarkable feature, which we will explore in a separate publication.

\subsubsection{1D Quantum walk with a constant coin on $\Ir_+$}
\label{sec:cmv_cc} As an example we will consider a quantum walk with a constant coin of the form \eqref{CMVcoin} with $\gamma_{2x}\equiv\gamma$
constant. The Schur parameters are $\gamma,0,\gamma,0,\gamma,0,\dots$ and hence the sequence of Schur iterates $f_k$ is periodic with period $L=2$
\cite{CMGVI,CGMV2010}. Therefore it satisfies a fixed point equation for a composition of $L$ steps of Schur iteration. In our case $f=f_2$ is just a
quadratic equation, which has the solution
\begin{equation}\label{ff2}
    f(z) = \frac{z^2-1+\sqrt{(z^2-1)^2+4|\gamma|^2z^2}}{2\overline\gamma z^2}.
\end{equation}
The branch of the square root is that one converging to 1 when $z\to0$.

The measure related to $f$, i.e. the spectral measure of the cyclic vector $|0,\uparrow\rangle$, has a weight and a mass point which is absent only
when the coin is symmetric \cite{CMGVI}. Hence, all the states are transient in case of a symmetric coin, otherwise a one-dimensional transient
subspace appears, namely, the eigenspace of the simple eigenvalue given by the mass point.

The return probability for the state $|0,\uparrow\rangle$ is
\[
 R = \|f\|^2  = \int_0^{2\pi} \frac{dt}{2\pi} |{f}(e^{it})|^2,
 \qquad
  {f}(e^{it}) = \frac{e^{-it}}{\overline\gamma} (h(t)+i\sin t),
\]
where
\[
h(t) = \begin{cases}
          \sign(\cos t) \; \sqrt{|\gamma|^2-\sin^2 t} & \text{ if } |\sin t| \leq |\gamma|,
          \\
          -\sign(\sin t) \; i\sqrt{\sin^2 t-|\gamma|^2} & \text{ if } |\sin t| > |\gamma|.
       \end{cases}
\]
The final result is
\begin{equation} \label{R-Z+}
R = \frac{2}{\pi|\gamma|^2} \left\{ \rho|\gamma| + (1-2\rho^2)\eta \right\}, \quad \rho=\sqrt{1-|\gamma|^2}, \quad \eta=\arcsin|\gamma|,
\end{equation}
which actually holds for any basis state due to Khrushchev's result and the fact that $f(z)=f_{2x}(z)=z^{-1}f_{2x-1}(z)$. The return probability does
not depend on the phase of $\gamma$, and it is an increasing function of $|\gamma|$.

Using the generating function of the Legendre polynomials $P_n$,
\[
\frac{1}{\sqrt{1-2xz+z^2}} = \sum_{n\geq0}P_n(x)z^n,
\]
we can also find the power expansion of $f$, obtaining for $|0,\uparrow\rangle$ the arrival amplitudes $a_n$, where ${a}_1=\overline\gamma$ and
\[
{a}_{2n}=0, \qquad {a}_{2n+1}=\frac{P_{n-1}(x)-xP_n(x)}{2\gamma(n+1)}, \qquad x=1-2|\gamma|^2, \qquad n\geq1.
\]
Bearing in mind that $P_n(\cos(\theta)) \sim \sqrt{2/\pi n \sin(\theta)} \cos[(n+1/2)\theta-\pi/4]$ we find that the first return
probabilities $|a_n|^2$ decay as $n^{-3}$,
\[
|a_{2n+1}|^2 \sim \frac{\rho}{2\pi|\gamma|} \{1-\sin[(4n+2)\eta]\} \, n^{-3}.
\]

\subsubsection{1D Quantum walk with a constant coin on $\mathbb{Z}$}
\label{sec:ccdouble} Following the approach in \cite{CMGVI,CGMV2010}, we choose the basis $\{|k\rangle\}_{k=0}^\infty = \{\ket{0,\uparrow}$,
$\ket{{-}1,\downarrow}$, $\ket{{-}1,\uparrow}$, $\ket{0,\downarrow}$, $\ket{1,\uparrow}$, $\ket{{-}2,\downarrow}$, $\ket{{-}2,\uparrow}$,
$\ket{1,\downarrow},\dots\}$ to mimic the half infinite situation.  This leads to a  CMV representation $\mathcal C$ of the unitary step operator
similar to \eqref{eq:walk_cmv} with vanishing odd Schur parameters, but {\it matrix valued} ones, i.e.,  $\gamma_k$ and $\rho_k$ are now antidiagonal,
respectively diagonal, $2\times2$ matrices. The Carath\'{e}odory function and the Schur function likewise become matrix valued. The final result is that,
for any basis state, the first return probabilities $|a_n|^2$ decay as $n^{-3}$ like in the previous case,
\[
   a_{2n-1}=0, \quad |a_{2n}|^2 \sim \frac{2\rho}{\pi|\gamma|} \{1+\sin(4n\eta)\} n^{-3}, \quad \rho=\sqrt{1-|\gamma|^2}, \quad \eta=\arcsin|\gamma|,
\]
and the total return probability is given in terms of the Schur function $f$ in \eqref{ff2},
\begin{equation}\label{returnIr}
  R = \|f^2\|^2  = \frac{2}{\pi|\gamma|^4} \left\{ (1+2\rho^2)\rho|\gamma| + (1-4\rho^2)\eta \right\}.
\end{equation}
We have a similar behavior as in the case of $\mathbb{Z}_+$ but, for the same coin, the return probability is slightly smaller for $\mathbb{Z}$ than
for $\mathbb{Z}_+$, as Figure \ref{fig:R} shows.

\begin{figure}[ht]
\centering
 \includegraphics[width=5cm,height=5cm]{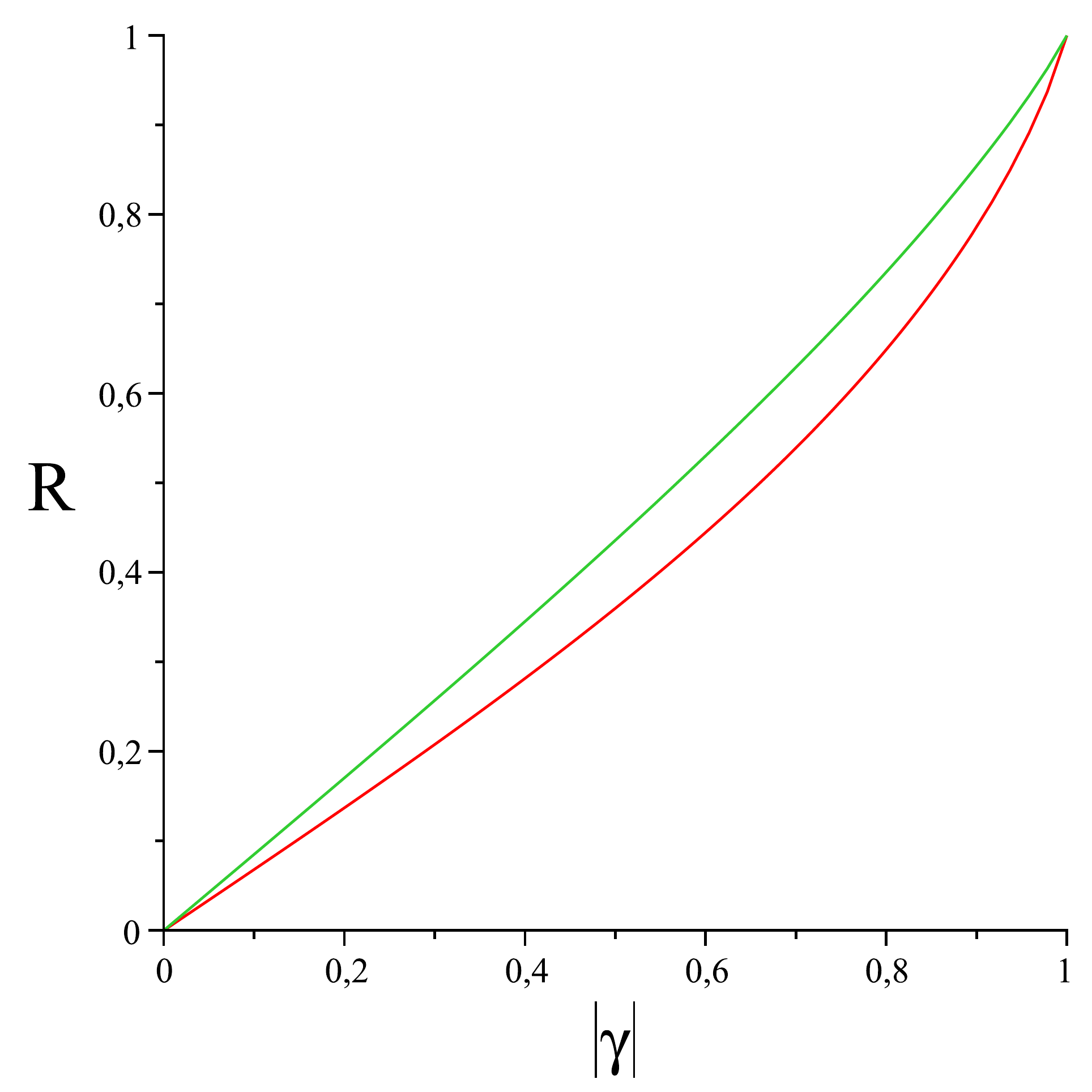}
  \caption{The return probability $R$ for a basic state $|k,\uparrow\rangle$, $|k,\downarrow\rangle$ of a quantum walk with a constant coin $C=(c_{ij})$
  as a function of $|\gamma|=|c_{21}|$. The upper curve corresponds to a constant coin on the non-negative integers, the lower one on the integers.}
  \label{fig:R}
\end{figure}

We can summarize stating that the recurrence properties for a constant coin on $\mathbb{Z}$ or $\mathbb{Z}_+$ are similar, but the recurrence is
slightly more prominent when the translation invariance is broken by the boundary conditions in $\mathbb{Z}_+$.

\subsection{Return times by the Fourier method}
This method was developed for translationally invariant quantum walks \cite{Scudo,timeRandom}. In this sense it is much more special than the CGMV
method described above, but has the advantage that it works just as well in higher lattice dimension, and with an arbitrary initial state, without
the need to recompute the Verblunsky coefficients. Moreover, the Fourier method has been extended to certain decoherent walks and interacting
two-particle systems \cite{timeRandom,spacetimeRandom,molecules}.

The Hilbert space is now $\WSp$, where $s\in\Nl$ is the dimension of the underlying spatial lattice, and $\KK$ is the space of internal degrees of
freedom at each site $x$. We only assume that $U$ is unitary on this Hilbert space, and commutes with the lattice translations. Often it is made part
of the definition of a quantum walk that it has strictly finite propagation in each step, but we will not need this condition here. Therefore, we can
jointly diagonalize $U$ with the translations by a Fourier transform in the spatial variable, which means that the Hilbert space vectors are
naturally considered as $\KK$-valued functions of the Fourier variable ``momentum'' $p=(p_1,\ldots,p_s)\in[-\pi,\pi]^s$. In this representation $U$
becomes multiplication by a $p$-dependent unitary operator on $\KK$, denoted by $U(p)$. The assumption of finite propagation speed would imply that
the matrix elements of $U(p)$ are Laurent polynomials in the variables $\exp(ip_k)$.

Then for any initial vector $\phi$  we get the Stieltjes function \eqref{muh}
\begin{equation} \label{eq_stielt_four}
  \muh(z) = \bra{\phi}{(1-z U)^{-1}}\ket{\phi} = \int dp\; \braket{\phi(p)}{(1-z U(p))^{-1}|\phi(p)}
\end{equation}
When the starting point is the origin, $\phi(p)\equiv\phi\in\KK$ is independent of $p$, and we can evaluate the integral as $\muh(z)=\bra\phi
M(z)\ket\phi$, with the Stieltjes operator
\begin{equation}\label{StieltjesM}
    M(z)=\int dp\; {(1-z U(p))^{-1}}.
\end{equation}

For a coined one-dimensional quantum walk with the same coin as in Sect.~\ref{sec:cmv_cc} the walk operator takes the form
\begin{align}\label{eq:def_qw}
 U(p)=S(p)\cdot C = \begin{pmatrix}
   e^{\ic p} &0\\  0 &e^{-\ic p} \end{pmatrix}
   \begin{pmatrix} \rho  & -\gamma \\ \overline\gamma &  \rho \end{pmatrix}.
\end{align}
The integral \eqref{StieltjesM} can be evaluated by the residue theorem and gives
\begin{align}
  M(z) =\frac{1}{2 g(z)} \begin{pmatrix}
     1-z^2+ g(z) &- \frac{\gamma}{\rho}  ( 1+z^2- g(z))\\
   \frac{\overline \gamma}{\rho} ( 1+z^2- g(z))& 1-z^2+ g(z)
  \end{pmatrix}\;.
\end{align}
with $g(x)=\sqrt{(1-z^2)^2+4\abs{\gamma}^2 z^2}$. From this we can immediately write down the Schur function and compute the total return probability
$R$ as the 2-norm. In agreement with Sect.~\ref{sec:ccdouble} the cases $\phi=\ket{\uparrow}$ and $\phi=\ket{\downarrow}$ give the same result
\eqref{returnIr}, because the diagonal elements of $M(z)$ are equal. However, for starting/absorbing states $\phi$ which are superpositions of these
two we get other Schur functions, because $M(z)$ is not a multiple of the identity.
\begin{figure}[ht]
\centering
 \includegraphics[width=6cm]{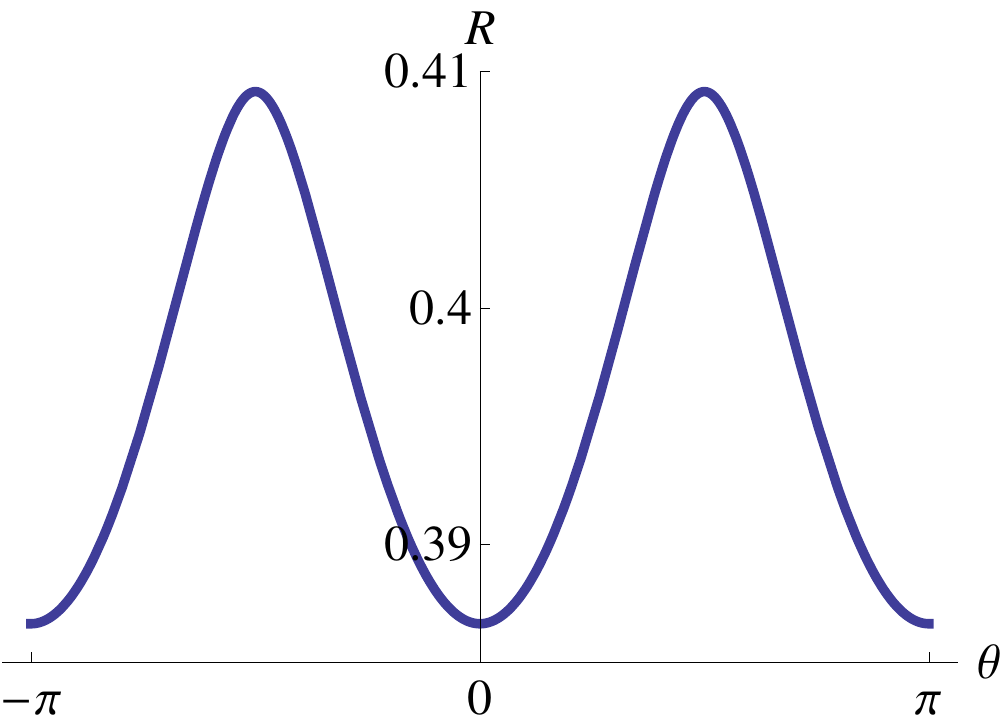}
  \caption{$\theta$-dependence of the return probability $R$ of a quantum walk with $\gamma = 1/\sqrt 2$ (see \eqref{eq:def_qw}) of the initial state $\phi=(1,\exp(\ic\theta))$.}
  \label{fig:RQW}
\end{figure}

Indeed, this can lead to different $R$. This can be seen from fig. \ref{fig:RQW}, which shows the dependence of $R$ on the relative phase of the initial state $\phi=(1,\exp(\ic\theta))$ for the quantum walk with  $\gamma = 1/\sqrt 2$ according to \eqref{eq:def_qw}.

\subsubsection{Local perturbations}
In the CGMV approach a local perturbation is easily implemented by letting the Schur coefficients be constant only from some index onwards. From
there on the $f_k$ are determined by the fixed point equation, and $f=f_0$ is obtained by inverse Schur iteration \eqref{SchurIterate}.

In the Fourier approach a local perturbation can be turned into a perturbation only at the origin, by redefining a block of cells as a single one.
The resolvent of the perturbed unitary operator can be evaluated with the resolvent formula for perturbations. Since the ``unperturbed'' resolvent of
the above operator $U$ can be computed separately for each $U(p)$ this leads to a finite dimensional, albeit $p$-dependent matrix computation. Of
course, as in the CGMV case this may quickly lead to rather unwieldy expressions. However, at least for numerical work one gets a systematic method.

%% file: recurS.tex
\section{Outlook}
\subsection{Examples with singular measures}
Obviously, there is a wide range of examples which one would like to treat under this heading. In the quantum walk field, disordered walks, i.e.,
operators with space dependent i.i.d.\ random coins come to mind. In this case one has Anderson localization \cite{disordered}, i.e., dense point
spectrum. In almost periodic examples, such as walks in an irrational magnetic field, one expects purely singular spectrum. For irrationals which are
well approximated by rationals, e.g., Liouville numbers, there may be late returns grouped in a hierarchy of time scales. A further family of
examples is given by measures with some self-similarity such as a famous example by Riesz \cite{Riesz1918,Simon1,GV}.For the time being we have no
general results about return times in these examples, but it seems to be a rich field for investigations linking spectral properties with the
dynamics of propagation, which we intend to explore further.

\subsection{Generalizations}
We have focused completely on the case where absorption happens only in one pure state, which is also the initial one. The main reason for this was
to have a very direct link to spectral theory, and a dependence only on the pair $(U,\phi)$ and nothing else. However, from the point of view of
physical applications some generalizations are desirable, and can be pursued very much along the lines set out in this paper. Indeed there is no
reason why in Sect.~\ref{sec:def} the initial and the absorbing state should be the same, so one can immediately generalize to ``first arrival''
rather than ``first return''. It is often natural to consider a more than one dimensional space to be absorbing, e.g., when we want to discuss the
return of a coined walk to the origin, regardless of the internal state is. In that case one might also want to look at spin resolved arrival events
and, quite generally, at an array of counters. This has been studied in the continuous time case \cite{Allcock,arrival}, which is a further
generalization of interest. When an appropriate weakening of the absorption process is introduced to avoid the Zeno effect \cite{Castrigiano} in the
continuous time limit one finds a picture analogous to the discrete time case: The absorbing perturbation then modifies a one-parameter unitary group
to a semigroup of contractions. The difference of their generators, sometimes called an optical potential, may be unbounded or even fail to be a
closable operator, as in the case of an absorbing point ``potential'' $-i\lambda\delta(x)$ perturbing the Laplacian (free particle). The interplay
between spectrum of the Hamiltonian and arrival time in the continuous case contains a recently established uncertainty relation \cite{Uncertainty},
but much is left to explore.